\documentclass[12pt]{article}
\pdfoutput=1
\usepackage[utf8]{inputenc}

\usepackage{amsmath, amssymb, latexsym, amsthm}
\usepackage{fullpage}
\usepackage{color}
\usepackage{tikz}
\usepackage{graphicx}
\usepackage{stmaryrd,amsbsy,enumerate}
\usepackage{hyperref}
\usepackage{mathscinet}
\usetikzlibrary{calc}
\usepackage{cancel}
\usepackage{xspace}
\usepackage{tabularx,etoolbox}
\usepackage{url}
\usepackage{longtable}
\usepackage[protrusion=true,expansion=true]{microtype}

\newtheorem{theorem}{Theorem} 
\newtheorem{theorem*}{Theorem} 
\newtheorem{proposition}[theorem]{Proposition}

\newtheorem{lemma}[theorem]{Lemma}

\theoremstyle{definition}

\newtheorem{observation}[theorem]{Observation}
\newtheorem{question}[theorem]{Question}

\theoremstyle{remark}

\newcommand{\FPT}[0]{\ensuremath{\mathrm{FPT}}\xspace}
\newcommand{\NP}[0]{\ensuremath{\mathrm{NP}}\xspace}
\newcommand{\ZPP}[0]{\ensuremath{\mathrm{ZPP}}\xspace}
\newcommand{\Pp}[0]{\ensuremath{\mathrm{P}}\xspace}
\newcommand{\XP}[0]{\ensuremath{\mathrm{XP}}\xspace}
\newcommand{\MSOo}[0]{\ensuremath{\mathrm{MSO}_1}\xspace}

\newcommand{\prob}[4]{
\begin{center}
\bgroup
\renewcommand{\arraystretch}{1.1}%
\begin{tabularx}{\textwidth}{|llXr|}
	\hline
	\hspace{3pt}\rule{0pt}{13pt}&\multicolumn{2}{l}{#1}&\\
	&{\bf Input:\enspace}&{#2}&\hspace*{3pt}\\
	&{\bf Question:\enspace}&{#3\rule[-6pt]{0pt}{6pt}}& \\
	\hline
\end{tabularx}
\egroup
\end{center}
}

\date{}
\begin{document}

\title{Notes on complexity of packing coloring}
\author{
Minki Kim\thanks{Department of Mathematical Sciences, KAIST, Daejeon, South Korea, E-mail: {\tt kmk90@kaist.ac.kr}.}  \and
Bernard Lidick\'{y}\thanks{Department of Mathematics, Iowa State University, Ames, IA, E-mail: {\tt lidicky@iastate.edu}. Research of this author is supported in part by NSF grant DMS-1600390.}
\and
Tom\'a\v{s} Masa\v{r}\'{\i}k\thanks{ Department of Applied Mathematics of the Faculty of Mathematics and Physics at the Charles University, Prague, Czech Republic. E-mail: {\tt masarik@kam.mff.cuni.cz} Research of this author is supported by the grant SVV–2017–260452 and by the project GA17-091425 of GA \v{C}R.}
\and
Florian Pfender\thanks{Department of Mathematical and Statistical Sciences, University of Colorado Denver, E-mail: {\tt 
Florian.Pfender@ucdenver.edu}. Research of this author is supported in part by NSF grant DMS-1600483.} 
}

\maketitle

\begin{abstract}
A packing $k$-coloring for some integer $k$ of a graph $G=(V,E)$ is a mapping
 $\varphi:V\to\{1,\ldots,k\}$ such that any two vertices $u, v$ of color $\varphi(u)=\varphi(v)$ are in distance at least $\varphi(u)+1$.
This concept is motivated by frequency assignment problems.
The \emph{packing chromatic number} of $G$ is the smallest $k$ such that there exists a packing $k$-coloring of $G$.

Fiala and Golovach showed that determining the packing chromatic number for
chordal graphs is \NP-complete for diameter exactly 5.
While the problem is easy to solve for diameter 2,
we show \NP-completeness for any diameter at least 3. 
Our reduction also shows that the packing chromatic number is hard to approximate within $n^{{1/2}-\varepsilon}$ for any $\varepsilon > 0$.

In addition, we design an \FPT algorithm for interval graphs of bounded diameter.
This leads us to exploring the problem of  finding a partial coloring
that maximizes the number of colored vertices.
\end{abstract}

\section{Introduction}

Given a graph $G=(V,E)$ and an integer $k$, a \emph{packing $k$-coloring} is a mapping $\varphi:V\to\{1,\ldots,k\}$ such that any two vertices $u, v$ of color $\varphi(u)=\varphi(v)$ are in distance at least $\varphi(u)+1$.
An equivalent way of defining the packing $k$-coloring of $G$ 
is that it is a partition of $V$ into sets $V_1,\ldots,V_k$ such that for all $k$ and any $u,v \in V_k$, the distance
between $u$ and $v$ is at least $k+1$.
The \emph{packing chromatic number} of $G$, denoted $\chi_P(G)$, is the smallest $k$ such there exists a packing $k$-coloring of $G$.

The definition of packing $k$-coloring is motivated by frequency assignment problems. 
It emphasizes the fact that the signal on different frequencies can travel different distances. 
In particular, lower frequencies, modeled by higher colors, travel further so they may be used less often than higher frequencies.
The packing coloring problem was introduced by Goddard et al.~\cite{GHHHR} under the name \emph{broadcasting chromatic number}. 
The term packing coloring was introduced by Bre\v{s}ar, Klav\v{z}ar, and Rall~\cite{Bresar}.

Determining the packing chromatic number is often difficult.
For example, Sloper~\cite{Sloper} showed that the packing chromatic number of  the infinite 3-regular tree is 7
but the infinite 4-regular tree does not admit any packing coloring by a finite number of colors.
Results of Bre\v{s}ar, Klav\v{z}ar, and Rall~\cite{Bresar}  and  Fiala, Klav\v{z}ar and Lidick\'y~\cite{FKL} imply that the packing chromatic number of the infinite hexagonal lattice is 7.

Looking at these examples, researchers asked the question if there exists a constant $p$ such that every subcubic graph has packing chromatic number bounded by $p$.
A very recent result of Balogh, Kostochka and Liu~\cite{Balogh} shows that there is no such $p$ in quite a strong sense.
They show that for every fixed $k$ and $g\geq 2k+2$, almost every $n$-vertex cubic graph of girth at least $g$ has packing chromatic number greater than $k$.
It is still open if a constant bound holds for planar subcubic graphs, and no deterministic construction of subcubic graphs with arbitrarily high packing chromatic number is known.

Despite a lot of effort~\cite{FKL,GHHHR,Barnaby,Soukal}, the packing chromatic number of the square grid is still not determined.
It is known to be between 13 and 15 due to  Barnaby, Franco, Taolue, and Jos~\cite{Barnaby},
who use state of the art SAT-solvers to tackle the problem.

In this paper, we consider the packing coloring problem from the computational complexity point of view.
In particular, we study the following problem.

\prob{\sc Packing $k$-coloring of a graph}{A graph $G$ and a positive integer $k$.}{Does $G$ allow a packing $k$-coloring?}{s}

\subsection{Known results}
We characterize our algorithmic parameterized results in terms of \FPT (running time $f(k) \text{poly}(n)$) and \XP (running time $n^{f(k)}$) where $n$ is the size of the input, $k$ is the parameter and $f$ is any computable function.

The investigation of computational complexity of packing coloring was started by Goddard et al.~\cite{GHHHR} in 2008. They showed that {\sc packing $k$-coloring}  is \NP-complete for general graphs and $k=4$ and it is polynomial time solvable for $k\leq3$.
Fiala and Golovach~\cite{FG10} showed that {\sc packing $k$-coloring} is \NP-complete for trees for large $k$ (dependent on the number of vertices). 

For a fixed $k$,  {\sc packing $k$-coloring}  is expressible in \MSOo logic.  
Thus, due to Courcelle's theorem~\cite{courcelle}, it admits a fixed parameter tractable (\FPT) algorithm parameterized by the treewidth or clique width~\cite{CMR} of the graph.
Moreover, it is solvable in polynomial time if both the treewidth and the diameter are bounded~\cite{FG10}. 
The problem remains in \FPT even if we fix the number of colors that can be used more than once by the extended framework of Courcelle, Makowsky and Rotics~\cite{CMR}, see Theorem~\ref{thm:CMR}.
On the other hand, the problem is \NP-complete for chordal graphs of diameter exactly 5~\cite{FG10},
 and it is polynomial time solvable for split graphs~\cite{GHHHR}. Note that split graphs are chordal and have diameter at most 3. 
 However, {\sc packing $k$-coloring} admits an \FPT algorithm on chordal graphs parameterized by $k$~\cite{FG10}.

\subsection{Our results and structure of the paper}
We split our results into two parts.
 
In Section~\ref{sec:Chord}, we describe new complexity results on chordal, interval and proper interval graphs.
We improve a result by Fiala and Golovach~\cite{FG10} to chordal graphs of any diameter greater or equal than three. Moreover, we imply an inapproximability result (Theorem~\ref{main}). Chordal graphs of diameter less than three are polynomial time solvable (Proposition~\ref{prop:poly}).
We complement these results by several \FPT and \XP algorithms on interval and proper interval graphs. We use dynamic programming as an \XP algorithm for interval graphs of bounded diameter (Theorems~\ref{thm:diameter}).
For unit interval graphs, there is an \FPT algorithm parameterized by the size of the largest clique (Theorem~\ref{unitint}). Note that the existence of an \FPT algorithm parameterized by path-width would imply an \FPT algorithm for general interval graphs parameterized by the size of the largest clique, but this remains unknown. We provide an \XP algorithm for interval graphs parameterized by the number of colors that can be used more than once (Theorem \ref{thm:dynamic}).

In Section~\ref{sec:param}, we describe complexity results and algorithms parameterized by structural parameters.
We design \FPT algorithms for them. For standard notation and terminology we refer to the recent book about parameterized complexity~\cite{param}.

The packing coloring problem is interesting only when the number of colors is not bounded. Otherwise, we can easily 
model the problem by a fixed \MSOo formula and use the \FPT algorithm by Courcelle~\cite{courcelle} parameterized by the clique 
width of the graph. We show that we can do a similar modeling even when we fix only the number of colors that can be used 
more than once and then use a stronger result by Courcelle, Makowski and Rotics~\cite{CMR} that gives an \FPT algorithm parameterized by clique width of the graph (Theorem~\ref{thm:CMR}).

If the number of such colors is part of the input, then we can solve the problem on several structural graph classes. If a structural graph class has bounded diameter, then we can use Theorem~\ref{thm:CMR} due to the following easy observation.

\begin{observation}\label{obs}
Let $G$ be a graph of bounded diameter. Then $G$ has a bounded number of colors that can be used more than once.
\end{observation}

This observation together with Theorem~\ref{thm:CMR} implies that the problem is \FPT for any class of graphs of bounded shrub depth. Any class of graphs that has bounded shrub depth has a bounded length of induced paths (\cite{sd}, Theorem 3.7) and thus bounded diameter. The same holds for graphs of bounded modular width as they have bounded diameter according to Observation~\ref{obs:mw}.
On the other hand, the problem was shown to be hard on graphs of bounded treewidth~\cite{FG10},  in fact the problem is \NP-hard even on trees. There seems to be a big gap and thus interesting question about parameterized complexity with respect to pathwidth of the graph. It still remains open (Question~\ref{que:pw}).
Note that the original hardness reduction by Fiala and Golovach~\cite{FG10} has unbounded pathwidth  since it contains large stars.

We refer~\cite{mw} for the definition of modular width and its construction operations.

\begin{observation}\label{obs:mw}
	Let $G$ be a graph of modular width $k$. Then $G$ has diameter at most $\max(k,2)$.
\end{observation}
\begin{proof}
	We look at the last step of the decomposition. It has to create a connected graph and thus it is either a join operation or a template operation.
	If it is the join operation then the diameter is at most $2$ and if it is the template operation the longest path between any two vertices in different operands is at most $k$ and if they are in the same operand their distance is at most $2$.
\end{proof}

See Figure~\ref{fig:classes} for an overview of the results with respect to the structural parameters.

\section{Chordal and Interval graphs}\label{sec:Chord}

\begin{proposition}\label{prop:poly}
Packing chromatic number is in $P$ for chordal graphs of diameter 2.
\end{proposition}
\begin{proof}
Let $G$ be a chordal graph of diameter $2$. 
Notice that in graphs of diameter $2$, the only color 
that can be used more than once is color 1.
Hence, determining the packing chromatic number of $G$ is equivalent to finding a largest independent set in $G$.
In chordal graphs, the larges independent set can be found in polynomial time.
Hence $\chi_P(G)$ can be found in polynomial time.
\end{proof}

For larger diameters, we use a similar reduction as Fiala and Golovach~\cite{FG10} to finding a largest independent set in a general graph. 
\ZPP is a complexity class of problems which can be solved in expected polynomial time by a probabilistic algorithm that never makes an error. 
It lies between $\Pp$ and $\NP$ ($\Pp\subseteq \ZPP \subseteq \NP$). 
It is strongly believed that $\ZPP\neq \NP$.
H\aa{}stad~\cite{Haastad1999} showed that finding a largest independent set  is  hard to approximate.
\begin{theorem}[H\aa{}stad \cite{Haastad1999}]\label{thm:hastad}
Unless $\NP = \ZPP$, Max-Clique cannot be approximated within $n^{1-\varepsilon}$ for any $\varepsilon >0$.
\end{theorem}
Together with our reduction, this implies that the packing chromatic number is hard to approximate.

\begin{theorem}\label{main}
Packing chromatic number is \NP-complete on chordal graphs of any diameter at least 3. Moreover, it is hard to approximate within $n^{1/2-\varepsilon}$ for any $\varepsilon >0$, unless $\NP = \ZPP$.
\end{theorem}
\begin{proof}
We use a reduction to the independent set  problem.
Let $G$ be any connected graph on $n$ vertices.
We construct a chordal graph $H$ of diameter $d\ge 3$ from $G$ by the following sequence of operations:
\begin{enumerate}[(a)]
\item start with $G$, denote the set of its vertices by $V$,
\item subdivide every edge once, denote the set of new vertices by $S$,
\item add all possible edges between vertices in $S$,
\item for every $v \in V$ add a duplicate vertex $v'$ and the edge $vv'$; denote the set of new duplicate vertices by $D$, 
\item to increase the diameter to $d>3$, add a path $P$ of length $d-2$ starting in one vertex in $S$. \label{addX}
\end{enumerate}
See Figure~\ref{reduction} for an example of the construction.
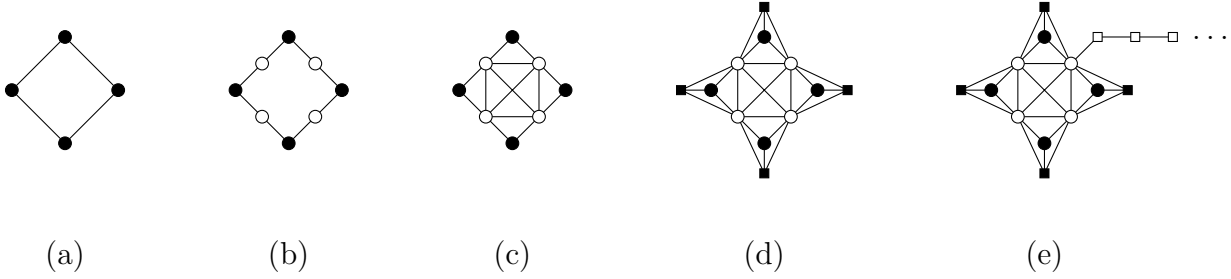
\begin{figure}
\begin{center}
\tikzset{vtx/.style={inner sep=1.7pt, outer sep=0pt, circle, draw,fill}} 
\tikzset{vtxS/.style={inner sep=1.7pt, outer sep=0pt, circle,draw,fill=white}} 
\tikzset{vtxD/.style={inner sep=1.7pt, outer sep=0pt, rectangle,draw,fill}} 
\tikzset{vtxP/.style={inner sep=1.7pt, outer sep=0pt, rectangle,draw,fill=white}} 
\begin{tikzpicture}
\draw(0,0) node[vtx](a){} 
(a) -- (45:1) node[vtx](b){}
(a) -- (135:1) node[vtx](c){}
(c) -- ++(45:1) node[vtx](d){}
(b)--(d);
\draw(0,-1.5) node {(a)};
\end{tikzpicture}
\hskip 3em
\begin{tikzpicture}
\draw(0,0) node[vtx](a){} 
(a) -- (45:1) node[vtx](b){}
(a) -- (135:1) node[vtx](c){}
(c) -- ++(45:1) node[vtx](d){}
(b)--(d)
($(a)!0.5!(b)$) node[vtxS](x1){}
($(a)!0.5!(c)$) node[vtxS](x2){}
($(c)!0.5!(d)$) node[vtxS](x3){}
($(b)!0.5!(d)$) node[vtxS](x4){}
;
\draw(0,-1.5) node {(b)};
\end{tikzpicture}
\hskip 3em
\begin{tikzpicture}
\draw(0,0) node[vtx](a){} 
(a) -- (45:1) node[vtx](b){}
(a) -- (135:1) node[vtx](c){}
(c) -- ++(45:1) node[vtx](d){}
(b)--(d)
($(a)!0.5!(b)$) node[vtxS](x1){}
($(a)!0.5!(c)$) node[vtxS](x2){}
($(c)!0.5!(d)$) node[vtxS](x3){}
($(b)!0.5!(d)$) node[vtxS](x4){}
(x1)--(x2)--(x3)--(x1)--(x4)--(x2)
(x3)--(x4)
;
\draw(0,-1.5) node {(c)};
\end{tikzpicture}
\hskip 3em
\begin{tikzpicture}
\draw(0,0) node[vtx](a){} 
(a) -- (45:1) node[vtx](b){}
(a) -- (135:1) node[vtx](c){}
(c) -- ++(45:1) node[vtx](d){}
(b)--(d)
($(a)!0.5!(b)$) node[vtxS](x1){}
($(a)!0.5!(c)$) node[vtxS](x2){}
($(c)!0.5!(d)$) node[vtxS](x3){}
($(b)!0.5!(d)$) node[vtxS](x4){}
(x1)--(x2)--(x3)--(x1)--(x4)--(x2)
(x3)--(x4)
(a) -- ++(270:0.4) node[vtxD](da){}
(b) -- ++(0:0.4) node[vtxD](db){}
(c) -- ++(180:0.4) node[vtxD](dc){}
(d) -- ++(90:0.4) node[vtxD](dd){}
(x1)--(da)--(x2)
(x1)--(db)--(x4)
(x2)--(dc)--(x3)
(x3)--(dd)--(x4)
;
\draw(0,-1.5) node {(d)};
\end{tikzpicture}
\hskip 3em
\begin{tikzpicture}
\draw(0,0) node[vtx](a){} 
(a) -- (45:1) node[vtx](b){}
(a) -- (135:1) node[vtx](c){}
(c) -- ++(45:1) node[vtx](d){}
(b)--(d)
($(a)!0.5!(b)$) node[vtxS](x1){}
($(a)!0.5!(c)$) node[vtxS](x2){}
($(c)!0.5!(d)$) node[vtxS](x3){}
($(b)!0.5!(d)$) node[vtxS](x4){}
(x1)--(x2)--(x3)--(x1)--(x4)--(x2)
(x3)--(x4)
(a) -- ++(270:0.4) node[vtxD](da){}
(b) -- ++(0:0.4) node[vtxD](db){}
(c) -- ++(180:0.4) node[vtxD](dc){}
(d) -- ++(90:0.4) node[vtxD](dd){}
(x1)--(da)--(x2)
(x1)--(db)--(x4)
(x2)--(dc)--(x3)
(x3)--(dd)--(x4)
(x4)--++(45:0.5) node[vtxP]{} -- ++(0.5,0)node[vtxP]{}  -- ++(0.5,0)node[vtxP]{} ++(0.5,0)node[]{$\hdots$}
;
\draw(0,-1.5) node {(e)};
\end{tikzpicture}
\end{center}
\caption{The reduction from Theorem~\ref{main} on a $4$-cycle.}\label{reduction}
\end{figure}

We will choose   a packing coloring $\varphi$ of $H$ with $\chi_P(H)$ colors.
Notice that the graph induced by $V \cup S \cup D$ has diameter at most three. 
Hence, only colors 1 and 2 can be used more than once on  $V \cup S \cup D$.
We call colors other than 1 and 2 \emph{unique}.
Notice that we can freely permute the unique colors. 
Pick $\varphi$ in a way to maximize the number of unique colors among vertices in $S$,
and subject to that, to maximize the number of vertices in $D$ colored  1.
We will show that $S$ has only vertices of unique colors and all vertices in $D$ are colored  1.

Suppose for the sake of contradiction that there is a vertex $s \in S$ colored  1 or 2.
Since $S$ is a clique, $s$ is the only vertex in $S$ with this color.
Let $u \in D \cup V$  be a neighbor of $s$ with a unique color. 
Such a vertex must exist since $s$ has four neighbors in $D \cup V$, and at most two can be colored by 1 and 2.
Observe that by the construction of $H$, the closed neighborhood $N[u]\subseteq N[s]$.
Thus, for every vertex $w\ne u$, the distance $d(w,u)\le d(w,s)$.
Hence, we can swap the colors on $s$ and $u$, contradicting the choice of $\varphi$.
Therefore, all vertices in $S$ have unique colors.

Now let $x\in D$ and let $v$ be its unique neighbor in $V$. If $v$ has color 1, we can swap the colors on $x$ and $v$, contradicting our choice of $\varphi$. Therefore, no vertices in $N(x)$ have color 1, and thus $x$ has color 1 by our choice of $\varphi$.

Since all vertices in $D$ are colored  1, no vertex in $V$ can be colored 1.
Minimizing the number of unique colors on $V$ is the same as maximizing the number of vertices colored 2. 
By the distance constraints in $H$, a subset of $V$ can be colored 2 in $H$ if and only if it is an independent set in $G$.
Therefore, the vertices colored $2$ in $V$ form a largest independent set in $G$. 

Recall that in order to increase the diameter of $H$, we added the path $P$ with one endpoint $s \in S$ in step \eqref{addX}.
Notice that $P$ can be colored by a pattern of four colors starting in  $s$: $\varphi(s),1,2,1,3,1,2,1,3,1\ldots$.
The existence of the path neither increases $\chi_P(H)$ nor influences the coloring $\varphi$ in $V \cup D \cup S$.

Finally, notice that $H$ has at most ${n\choose 2}+2n+d-2$ vertices. 
Hence, if we could approximate  $\chi_P(H)$ with precision $(n^2)^{1/2-\varepsilon}$ for some $\varepsilon >0$, we could
approximate largest independent set in $G$ with precision $n^{1-2\varepsilon}$, which contradicts
Theorem~\ref{thm:hastad}.
\end{proof}

\begin{theorem}\label{thm:diameter}
Packing chromatic number for interval graphs of diameter at most $d$ can be solved in time $O(n^{d\ln (5d)})$.
\end{theorem}
\begin{proof}
Let $\varphi$ be a packing coloring of an interval graph $G$ with diameter $d$, and let $P$ be a diameter path in $G$.
Note that every interval corresponding to a vertex of $G$ intersects an interval corresponding to an internal vertex of $P$.
Suppose $X$ is a set colored by color $c \geq 2$ in $\varphi$.
Let $x_1,x_2\in X$, and let $p_1,p_2\in V(P)$ such that $x_1p_1,x_2p_2\in E(G)$. Then the distance between $p_1$ and $p_2$ is at least $c-1$.
Therefore, $|X|\le \frac{d-2}{c-1}+1$.

Therefore, only colors $1,\ldots,d-1$ can be used more than once by $\varphi$.
Notice that the number of vertices colored by $2,\ldots,d-1$ is upper bounded by 
\[
f(d) = \sum_{2\le c\le d-1}\left(\frac{d-2}{c-1}+1\right) =(d-2)(1+H(d-2))<d\ln (5d)-1,
\]
where $H(n)$ is the harmonic number.
There are at most $n^{f(d)}$ such partial colorings of $G$ by colors $2,\ldots,d-1$.
Finally, vertices colored by 1 form an independent set. 
Therefore, the following is an algorithm to find the packing chromatic number of $G$.

Enumerate all  $n^{f(d)}$ partial colorings by colors $2,\ldots,d-1$. 
For each partial coloring, find a maximum independent set in the remaining graph, which takes time $O(n)$ and color the remaining vertices with unique colors.
The whole algorithm runs in time $O(n^{f(d)+1}) = O(n^{d\ln (5d)})$.
\end{proof}

When restricting the class of graphs to unit interval graphs, we can find an FPT algorithm parametrized by the size of the largest clique, independent of diameter.
We need the following two results. 

\begin{lemma}[Goddard et al.~\cite{GHHHR}]\label{lem:path}
For every $s\in \mathbb{N}$, the infinite path can be colored by colors $s,s+1,\ldots,3s+2$.
\end{lemma}

\begin{proposition}[Fiala and Golovach~\cite{FG10}]\label{chordfpt}
Chordal graphs admit an \FPT algorithm parameterized by the number of colors used in the solution. 
\end{proposition}

\begin{theorem}\label{unitint}
Packing chromatic number for unit interval graphs with a largest clique of size at most $k$ is \FPT in $k$.
\end{theorem}
\begin{proof}
Let $G$ be a unit interval graph. As $G$ is perfect, we can find a partition of its vertex set into $k$ independence sets $X_1,\ldots,X_k$ in polynomial time.
Let $X_\ell = \{ v_1,v_2,\ldots,v_{|X_\ell|} \}$, where the $v_i$ are ordered corresponding to their interval representation. Note that for all $i < j$, 
the distance of $v_i$ and $v_j$ in $G$ is at least $j-i$. This implies that any packing coloring of a path on $|X_\ell|$ vertices can be used to packing color the set $X_\ell$ without conflicts.

Use Lemma~\ref{lem:path} to color each $X_\ell$ with colors $\{\frac52 (3^{\ell-1}-1)+1,\ldots, \frac52 (3^{\ell}-1)\}$, and notice that these color sets are disjoint. This yields a packing coloring of $G$ with at most $\frac52 (3^{k}-1)$ colors.
Therefore, the number of colors we need is bounded in terms of $k$, and we can apply
Theorem~\ref{chordfpt} to conclude the proof.
\end{proof}

In the previous argument, we saw that restricting the number of colors makes the problem simpler. While we obviously do not have such a restriction for all interval graphs, we can still achieve a result about partial packing colorings with a bounded number of colors along similar ideas.

\begin{theorem}\label{thm:dynamic}
Let $k$ be fixed and $G$ be an interval graph. Finding a partial coloring by colors $1,\ldots,k$ that is maximizing the number of colored vertices can be solved in time  $O(n^{k+2})$.
\end{theorem}
\begin{proof}
We compute a function $H(u_1,\ldots,u_k) \to \mathbb{N}$, which counts the maximum number of colored vertices such that $u_i$ has its interval with the right end-point most to the right among all vertices colored by color $i$. 
The domain of $H$ is $(V\cup \{N\})^k$, where $N$ is a symbol representing that a particular color was not used at all.
It is possible to compute $H$ using dynamic programming in time $O(n^{k+2})$.
\end{proof}
Notice that Theorem~\ref{thm:dynamic} implies Theorem~\ref{thm:diameter} with a smaller exponent in the running time.

\section{Structural parameters}\label{sec:param}

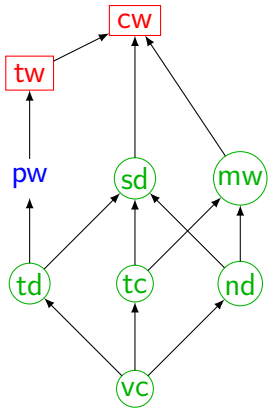
\begin{figure}[ht]
  \begin{minipage}[c]{0.35\textwidth}
\begin{center}
\tikzset{vtxR/.style={rectangle, draw,color=red, inner sep=3pt, outer sep=0pt}} 
\tikzset{vtxG/.style={circle, draw,color={green!70!black}, inner sep=0.8pt, outer sep=0pt}} 
\tikzset{vtxB/.style={color=blue, inner sep=0.5pt, outer sep=3pt}} 
\begin{tikzpicture}[scale=1.4]
\draw 
(0,0) node[vtxG](vc){$\mathsf{vc}$}
(1,1) node[vtxG](nd){$\mathsf{nd}$}
(1,2) node[vtxG](mw){$\mathsf{mw}$}
(-1,1) node[vtxG](td){$\mathsf{td}$}
(-1,2) node[vtxB](pw){$\mathsf{pw}$}
(-1,3) node[vtxR](tw){$\mathsf{tw}$}
(0,3.5) node[vtxR](cw){$\mathsf{cw}$}
(0,1) node[vtxG](tc){$\mathsf{tc}$}
(0,2) node[vtxG](sd){$\mathsf{sd}$}
;
\draw[-latex](vc)--(nd);
\draw[-latex](nd)--(mw);
\draw[-latex](vc)--(tc);
\draw[-latex](mw)--(cw);
\draw[-latex](tc)--(sd);
\draw[-latex](nd)--(sd);
\draw[-latex](tc)--(mw);
\draw[-latex](sd)--(cw);
\draw[-latex](vc)--(td);
\draw[-latex](td)--(sd);
\draw[-latex](td)--(pw);
\draw[-latex](pw)--(tw);
\draw[-latex](tw)--(cw);
\end{tikzpicture}
\end{center}
  \end{minipage}\hfill
  \begin{minipage}[c]{0.63\textwidth}
\caption{Hierarchy of graph parameters. An arrow indicates that a graph parameter upper-bounds the other. Thus, hardness results are implied in direction of arrows and algorithms are implied in the reverse direction.
Green circles and red rectangle colors distinguish between hardness results and \FPT algorithms provided. 
Blue color without boundary denotes that the hardness is unknown.
($\mathsf{cw}$ is clique width, $\mathsf{nd}$ is neighborhood diversity, $\mathsf{mw}$ is modular width, $\mathsf{pw}$ is path width, $\mathsf{sd}$ is shrub depth, $\mathsf{tc}$ is twin cover, $\mathsf{td}$ is tree depth, $\mathsf{tw}$ is tree width, $\mathsf{vc}$ is vertex cover.
See~\cite{param} for definitions.)
}
\end{minipage}
\label{fig:classes}
\end{figure}

\begin{theorem}\label{thm:CMR}
Let $k$ be fixed and $G=(V,E)$ be a graph of clique width $q$. Finding a partial packing coloring by colors $1,\ldots,k$ that is maximizing the number of colored vertices can be solved in \FPT time parameterized by $q$.
\end{theorem}

\begin{proof}
We model the problem as an extended formulation in \MSOo logic with one free variable $X$ that represents the large colors. We use a result by Courcelle, Makowski and Rotics~\cite{CMR} to solve this formula $\varphi(X)$ on graphs of clique width $q$ in \FPT time such that it minimize the size of the set $X$.

$$\varphi(X)\models  \exists X_1,\dots X_k\subseteq V \text{ s.t. } \forall i\ \text{$i$-independent}(X_i) \wedge V=X\dot\cup X_1 \dot\cup\cdots\dot\cup X_k. $$
$$\text{$i$-independent}(X)\models   \forall x,y\in X  \  d(x,y)\ge i.$$

$$d(x,y)\ge i\models \nexists z_1,\ldots,z_{i-1}\in V \text{ s.t. } x=z_1 \wedge y=z_{i-1}\wedge  \cup_{j=1}^{i-2}(\text{edge}(z_j,z_{j+1})\vee (z_j=z_{j+1})).$$
\end{proof}

\section{Conclusion}
Although the diameter is a widely investigated structural parameter we found that in some cases a related parameter better captures the problem, namely the number of colors that can be used more than once, as we show in Theorem~\ref{thm:dynamic}.

We close with a few open questions.
\begin{question}
Is determining the packing chromatic number for (unit) interval graphs in $\Pp$ or is it $\NP$-hard? 
\end{question}

\begin{question}
Is determining the packing chromatic number for interval graphs \FPT when parametrized by the largest clique size? 
\end{question}

One can think of graphs of bounded path-width as a generalization of interval graphs with bounded clique size. 
This leads to the following question.
\begin{question}\label{que:pw}
Is determining the packing chromatic number \FPT or \XP when parametrized by the path width?
\end{question}

Notice that Theorem~\ref{unitint} could be modified to work on graphs of bounded path width that have a decomposition such that every vertex is in a bounded number of bags.

\noindent\paragraph{Acknowledgements} 
All authors were supported in part by NSF-DMS Grants \#1604458, \#1604773, \#1604697 and \#1603823,``Collaborative Research: Rocky Mountain - Great Plains Graduate Research Workshops in Combinatorics''.

\bibliographystyle{abbrv}
\bibliography{refs.bib}

\end{document}